\documentclass{aptpub}
\usepackage[utf8]{inputenc}
\usepackage[english]{babel}
\usepackage[round]{natbib}
\usepackage{amsmath}
\usepackage{amssymb}
\usepackage{xcolor}
\usepackage{cleveref}

\usepackage{tikz}
\usepackage{pgfplots}

\makeatletter
\let\save@mathaccent\mathaccent
\newcommand*\if@single[3]{%
  \setbox0\hbox{${\mathaccent"0362{#1}}^H$}%
  \setbox2\hbox{${\mathaccent"0362{\kern0pt#1}}^H$}%
  \ifdim\ht0=\ht2 #3\else #2\fi
  }
\newcommand*\rel@kern[1]{\kern#1\dimexpr\macc@kerna}
\newcommand*\widebar[1]{\@ifnextchar^{{\wide@bar{#1}{0}}}{\wide@bar{#1}{1}}}
\newcommand*\wide@bar[2]{\if@single{#1}{\wide@bar@{#1}{#2}{1}}{\wide@bar@{#1}{#2}{2}}}
\newcommand*\wide@bar@[3]{%
  \begingroup
  \def\mathaccent##1##2{%
    \let\mathaccent\save@mathaccent
    \if#32 \let\macc@nucleus\first@char \fi
    \setbox\z@\hbox{$\macc@style{\macc@nucleus}_{}$}%
    \setbox\tw@\hbox{$\macc@style{\macc@nucleus}{}_{}$}%
    \dimen@\wd\tw@
    \advance\dimen@-\wd\z@
    \divide\dimen@ 3
    \@tempdima\wd\tw@
    \advance\@tempdima-\scriptspace
    \divide\@tempdima 10
    \advance\dimen@-\@tempdima
    \ifdim\dimen@>\z@ \dimen@0pt\fi
    \rel@kern{0.6}\kern-\dimen@
    \if#31
      \overline{\rel@kern{-0.6}\kern\dimen@\macc@nucleus\rel@kern{0.4}\kern\dimen@}%
      \advance\dimen@0.4\dimexpr\macc@kerna
      \let\final@kern#2%
      \ifdim\dimen@<\z@ \let\final@kern1\fi
      \if\final@kern1 \kern-\dimen@\fi
    \else
      \overline{\rel@kern{-0.6}\kern\dimen@#1}%
    \fi
  }%
  \macc@depth\@ne
  \let\math@bgroup\@empty \let\math@egroup\macc@set@skewchar
  \mathsurround\z@ \frozen@everymath{\mathgroup\macc@group\relax}%
  \macc@set@skewchar\relax
  \let\mathaccentV\macc@nested@a
  \if#31
    \macc@nested@a\relax111{#1}%
  \else
    \def\gobble@till@marker##1\endmarker{}%
    \futurelet\first@char\gobble@till@marker#1\endmarker
    \ifcat\noexpand\first@char A\else
      \def\first@char{}%
    \fi
    \macc@nested@a\relax111{\first@char}%
  \fi
  \endgroup
}
\makeatother

\theoremstyle{plain}
\date{}

\newcommand{\opd}{\, \mathrm{d}}

\begin{document}

\title{Linear prediction of point process times and marks}
%\author{Maximilian Aigner and Val\'{e}rie Chavez-Demoulin\\University of Lausanne}

\shorttitle{Prediction of point process}
\authornames{M.~AIGNER AND V.~CHAVEZ-DEMOULIN}

\authorone[University of Lausanne]{Maximilian Aigner}
\authortwo[University of Lausanne]{Valerie Chavez-Demoulin}

\address{Department of Operations, University of Lausanne, 1015 Lausanne, Switzerland}
\addresstwo{Department of Operations, University of Lausanne, 1015 Lausanne, Switzerland}

\emailone{maximilian.aigner@unil.ch}

\begin{abstract}
%\maketitle

%\begin{abstract}

%Linear prediction of a stationary stochastic process is a classical problem in time series analysis and spatial statistics. In essence, it relies on the stationarity of the underlying process to derive prediction equations for future times or unseen locations. 

In this paper, we are interested in linear prediction of a particular kind of stochastic process, namely a marked temporal point process. The observations are event times recorded on the real line, with marks attached to each event. We show that in this case, linear prediction extends straightforwardly from the theory of prediction for stationary stochastic processes. Following classical lines, we derive a Wiener-Hopf-type integral equation to characterise the linear predictor, extending the ``model independent origin'' of the Hawkes process \citep{jaisson_market_2015} as a corollary.

We propose two recursive methods to solve the linear prediction problem and show that these are computationally efficient in known cases. The first solves the Wiener-Hopf equation via a set of differential equations. It is particularly well-adapted to autoregressive processes. In the second method, we develop an innovations algorithm tailored for moving-average processes. A small simulation study on two typical examples shows the application of numerical schemes for estimation of a Hawkes process intensity.

\end{abstract}


\section{Introduction} 

\subsection{Marked point processes}

We begin by reviewing the main properties associated with a marked point process on the real line%. The stationary moment measures are not specific to temporal point processes, which are often characterised by the conditional intensity function instead. Both types of properties are needed for linear prediction.

%\subsection{Definitions}

Suppose $N$ is a marked point process on $\mathbb{R}$ with mark space $\left(\mathbf{M}, \mathcal{M}\right)$. The $i$-th event in the process $N$ occurs at a random time $T_i \geqslant 0$, and has an associated random mark $M_i \in \mathbf{M}$. We have a Dirac representation valid for $A \subseteq \mathbb{R}$ and $B \in \mathcal{M}$,
\begin{equation}
N(A \times B) = \sum_{i=1}^\infty \delta_{(T_i, M_i)}\left(A \times B\right),
\end{equation}
where the random times are an increasing sequence, $T_i \leq T_{i+1}$ for $i \in \mathbb{N}$. In this paper we restrict ourselves to simple processes for which no two times coincide, so that $T_i < T_{i+1}$.

Let $\mathcal{H}_t$ be the joint internal history of times and marks up to but not including time $t$, that is, 
$$
\mathcal{H}_t = \sigma\left(\left\{N((a, b) \times B \Big\vert a,b < t, B \in \mathcal{M}\right\}\right)
$$
The internal history process $\mathcal{H}_t$ represents the observations about $N$ that are known at time $t$.
We suppose that $N$ admits an $\mathcal{H}_t$-compensator, that is, a stochastic process $\Lambda$ indexed by $\mathbb{R} \times \mathbf{M}$ verifying
\begin{equation}
\mathbb{E} \bigg[ N\Big((a,b) \times M\Big) \bigg\vert \mathcal{H}_a \bigg] = \Lambda\big((a,b) \times M\big).
\label{eq:compensator_def}
\end{equation}
for every $a<b$ and $M \in \mathcal{M}$.
The above formula defines $\Lambda$ as a conditional mean measure of the process. It is a ``smoothing'' formula in the sense that integrals with respect to the pure jump process $N$ can be replaced, in expectation, by integrals with respect to the (usually) smoother process $\Lambda$.

In the simplest case, when $N$ is an unmarked ($\mathbf{M} = \left\{ 1 \right\}$) Poisson process with deterministic intensity $\lambda(t)$, we have simply $\Lambda(t) = \int_0^t \lambda(s) \mathrm{d}s$. 
More generally, if $\Lambda$ is absolutely continuous in its first argument with respect to Lebesgue measure, then it admits a density $\lambda(t, d\mu)$, and for any $M \in \mathcal{M}$ the equation \eqref{eq:compensator_def} implies the simpler formula
\begin{equation}
\mathbb{E} \bigg[ N\Big( (a, b) \times M\Big) \bigg| \,\mathcal{H}_a \bigg] = \int_a^b \int_M \lambda(t, \mathrm{d}\mu) \opd t.
\label{eq:compensator_with_density}
\end{equation}
where $\lambda(t, \mu)$, the conditional intensity, is a random function of $\mathcal{H}_t$ for every $\mu$.

If we disregard the marks, what remains is the ground process $\widebar{N} (t) = N\left([0, t) \times \mathbf{M}\right)$, for which a conditional intensity is given by
$$
\bar\lambda(t) = \int_\mathbf{M} \lambda(t, \mathrm{d}\mu).
$$
\begin{remark} In many applications the marked conditional intensity is taken to be separable, namely
$$
\lambda(t, \mathrm{d}\mu) = \bar\lambda(t) F_t(\mathrm{d}\mu)
$$
where $F_t$ is a fixed (not random) measure for every $t$. In this situation the past marks can influence the conditional intensity in the present, but the future marks depend only on time and not on the process history. The simplest case has $F_t \equiv F$ for every $t$, in which case marks are i.i.d.~according to $F$.
\end{remark}

%\subsection{Moment measures and stationarity}

% TODO: Tighten this, remove whitespace.
In the sequel, we will use second-order properties of the point process as well as the concept of stationarity. We briefly present these notions in our context. Suppose that $N$ is a square-integrable point process so that
$$
\mathbb{E}\left[N(t,B)\right], \mathbb{E}\left[N(t,B)^2\right] < +\infty.
$$

The mean or first moment measure of $N$ is the unique deterministic measure $M_1$ on $ \mathbb{R} \times \mathbf{M}$ which satisfies
\begin{equation}
\mathbb{E}\left[ \iint  f(t, \mu) N(\mathrm{d}t, \mathrm{d}\mu)\right] = \iint  f(t, \mu)\, M_1(\mathrm{d} t, \mathrm{d} \mu)
\label{eq:m1_measure}
\end{equation}
for all functions $f$ which are continuous with compact support. A shorthand version of \eqref{eq:m1_measure} is written
$
\mathbb{E}\left[N(\mathrm{d} t, \mathrm{d} \mu)\right] = M_1(\mathrm{d} t, \mathrm{d} \mu).
$
The second-order measure $M_2$ is the unique measure satisfying
%\begin{align*}
%&\mathbb{E}\left[ \iint \iint  f_1(t_1, \mu_1) f_2(t_2, \mu_2) N(\mathrm{d}t_1, \mathrm{d}\mu_1) N(\mathrm{d}t_2, \mathrm{d}\mu_2)\right]\\
%& = \iint \iint f_1(t_1, \mu_1) f_2(t_2, \mu_2)\, M_2(\mathrm{d}t_1, \mathrm{d}t_2, \mathrm{d}\mu_1, \mathrm{d}\mu_2)
%\end{align*}
%for every continuous $f_1$ and $f_2$ with compact support. Again, we may write this
$$
\mathbb{E}\left[N(\mathrm{d}t_1, \mathrm{d}\mu_1)N(\mathrm{d}t_2, \mathrm{d}\mu_2)\right] = M_2(\mathrm{d}t_1, \mathrm{d}t_2, \mathrm{d}\mu_1, \mathrm{d}\mu_2)
$$
which is made formal by integrating both sides against the product of two continuous functions with compact support. Finally, the covariance measure is defined as
$$
C(\mathrm{d}t_1, \mathrm{d}t_2, \mathrm{d}\mu_1, \mathrm{d}\mu_2) = M_2(\mathrm{d}t_1, \mathrm{d}t_2, \mathrm{d}\mu_1, \mathrm{d}\mu_2) - M_1(\mathrm{d}t_1,\mathrm{d}\mu_1) M_1(\mathrm{d}t_2, \mathrm{d}\mu_2).
$$
which can be made formal in the same way.

Each of the measures $M_1$ and $M_2$ represent expected behaviour of the pure jump process $N$, and typically have less jumps. In particular, they may admit densities with respect to time which we denote $m_1, m_2$:
\begin{eqnarray*}
M_1(\mathrm{d}t, B) = m_1(t, B)\, \mathrm{d}t  & M_2(\mathrm{d}t, \mathrm{d}s, B_1, B_2) =  m_2(t, s, B_1, B_2) \, \mathrm{d}t \,\mathrm{d}s
\end{eqnarray*}
We also define $c_2$ as the density of $C_2$, when this exists.

For the task of prediction, it is important that certain characteristics of the process remain the same throughout time. A strong form of stationarity supposes that the law of $N$ is invariant to shifts in time; this is usually impracticable. The concept of weak stationarity requires that the first and second moment measures do not change across time, and in particular are invariant to shifts in time. Thus, the process $N$ is said to be weakly stationary if the following two conditions hold \citep{daley_introduction_2003-1}
\begin{align} 
% TODO: Do I need the *reduced* moment measure ? Otherwise must explain the atom at the origin ! (DVJ §8.1)
M_1(\mathrm{d}t, \mathrm{d}\mu) &= \widebar{\lambda} \opd t\, F(\mathrm{d}\mu) \label{eq:stationary_m1}\\
M_2(\mathrm{d}t_1, \mathrm{d}t_2, \mathrm{d}\mu_1, \mathrm{d}\mu_2) &= \mathrm{d}t_1 M_2(\mathrm{d}(t_2 - t_1), \mathrm{d}\mu_1, \mathrm{d}\mu_2) \label{eq:stationary_m2}
\end{align}
for some real number $\widebar{\lambda} > 0$ and measure $F$ called the stationary mark distribution. Equation \eqref{eq:stationary_m2} means that the second-order measure (or the covariance measure) evaluated at two points depends on time only via the difference of both times. In the unmarked case where $\mathbf{M}$ reduces to a single point, these relations reduce to the stationarity conditions for a continuous-time stochastic process. Note, however, that we make no assumption about the last two arguments of $M_2$.

%\subsection{Hawkes process}
\paragraph{Example: the Hawkes process.}
One process in particular plays an important role in applications, namely the marked Hawkes process \citep{daley_introduction_2003, zugec_marked_2019}. This is a branching process on $\mathbb{R} \times \mathbf{M}$ composed of a marked Poisson process of immigrants, each of which initiates a family of descendants. The immigrants and all their descendants generate offspring according to an inhomogeneous marked Poisson process. It can be shown that such a process has the conditional intensity
\begin{equation}
\lambda(t, y) = \eta + \int_{-\infty}^t \int_\mathbf{M} K(t-s, \nu, y) N(\mathrm{d}s, \mathrm{d\nu}).
\label{eq:marked_hawkes}
\end{equation}
where $K(h, \nu, B)$ is the inhomogeneous Poisson process intensity for points with marks in $B$ and at temporal lag $h$ from their parent, whose mark has value $\nu$.
In \eqref{eq:marked_hawkes} the intensity at any time $t$ is influenced by past times and marks through the kernel $K$, which verifies three properties: $K\geq 0$, ensuring that the intensity remains positive; $K(h, x, y) = 0$ if $h < 0$, so that there is no dependence on the future; and $\mathbb{E} \int_0^{+\infty} K(h,x,y) \opd h < 1$ for every $B$, a stability property which ensures that the process of generations does not diverge by reaching infinite size in finite time.
%\end{example}

%\textcolor{red}{Transition and table of contents}

\subsection{Mean-square prediction}

In this section, we describe a few properties of mean-square prediction. These are not exactly specific to marked point processes but rather derive from properties of conditional expectation.% The developments which follow are a straightforward adaptation of the properties of conditional expectation.

% Paragraph introducing mean-square prediction

% Suppose we observe the process $N$ up but not including to some time $t^-$. The observations take the form of the internal history $\mathcal{H}_t$. We are interested in predicting future events, either at time $t$ or more generally in an interval $[t,  t+h)$ after $t$.

% In predicting $N$, we aim to minimise a loss function, which we take to be the mean square error for several reasons. It is true that the process $N$ takes only integer values, for which the mean-square error is perhaps not best suited. 

% The predictable part of $N$ is represented by the compensator $\Lambda$, whereas the martingale process $M$ represents unpredictable variation.

% Paragraph linking to time series theory

Suppose observations are given from a stationary square-integrable marked point process $N$ over the time interval $[0, s)$. The aim of this paper consists of predicting the path of the point process $N$ up to some future time $t$. The prediction is a stochastic process which uses the information of times and marks up to (but not including) a time $s \leqslant t$. 
More precisely, define the predictable $\sigma$-algebra $\mathcal{P}(\mathcal{H}_s)$ as that generated by sets of the form
$$
(u, s] \times A \times B, \text{ with } u \leqslant s,\, A \in \mathcal{H}_u \text{ and } B \in \mathcal{M}.
$$
A stochastic process $\xi$ is then called $\mathcal{H}_t$-predictable if $\xi$ is measurable with respect to $\mathcal{P}(\mathcal{H}_t)$, for every $t$. Roughly speaking, predictable processes are almost equivalent to left-continuous processes whose value $\xi(t, \cdot)$ is measurable with respect to $\mathcal{H}_t$.

\begin{remark}
In the setting of prediction from time $s$ up to time $t$, we will suppose that no new information arrives after point $s$, that is, the prediction $\xi(t)$ is measurable with respect to $\mathcal{P}(\mathcal{H}_{s})$ for all $t \geqslant s$. 
\end{remark}

%\begin{remark}
%The instantaneous case $s = t$ is perhaps the most common.  
%\end{remark}

%\begin{remark}
%A point process $N(t, B)$ is typically not predictable with respect to its own history, 
%\end{remark}

\begin{remark}
These definitions lead to another way of viewing the $\mathcal{H}_t$-compensator $\Lambda$: it is the $\mathcal{H}_t$-predictable projection \citep{jacod_multivariate_1975} of $N$ onto $\mathcal{H}_t$, consisting of that part of the process $N$ that can (in principle) be predicted from the past. The innovation martingale $\varepsilon_N = N - \Lambda$ is a component which is unpredictable given the history.
\end{remark}

To evaluate the accuracy of a predictor for $N$, we use the mean-square error for its attractive mathematical properties. We then turn to finding the best linear predictor in the sense of minimising the squared error. This setting is exactly similar to that of minimising prediction error of a time series \citep{brockwell_time_2006}.
%Among all possibilities for predicting $N$, we opt for the choice of the best predictor of $N$ will depend on the criterion used to measure accuracy. In this paper, the predictor will be evaluated using the mean-square error (MSE) because it allows explicit and computationally tractable formulas. %To use this metric, we must make the following hypothesis.

To sum up, in this work we focus on predicting $N\big((s, t), B\big)$, where $B \in \mathcal{M}$ is a set of marks, for arbitrary $t \geqslant s$ based on observation of $N$ up to time $s$, in the mean-square sense. This problem can be formulated as follows: 
\begin{align}
\text{ minimise} \quad & 
\mathbb{E}\left[ \Big(N\big((s, t), B\big) - \widehat{N}\big((s, t), B\big) \Big)^2 \right] \label{prob:min_square} \tag{BP}\\
\text{with respect to} \quad & \widehat{N} \in \mathcal{P}(\mathcal{H}_s)\nonumber
\end{align}

It is clear that in the absence of other constraints, the minimum of \eqref{prob:min_square} is achieved by the $\mathcal{H}_s$-compensator of $N$. This is a consequence of the representation \eqref{eq:compensator_def} for the compensator as a conditional expectation and the classical fact that the conditional expectation is a minimiser of the mean square error. Thus, if we denote $N^*$ the minimiser of \eqref{prob:min_square}, $N^*\big((s, t), B\big) = \mathbb{E}\left[ N\big((s, t), B\big) \Big\vert \mathcal{H}_s\right] 
= \Lambda\big((s, t), B\big)$
is the best overall predictor of $N\big((s,t), B\big)$. In fact, for any other predictor $\widehat{N}$ of $N$, we have
\begin{align}
&\mathbb{E}\left[ \Big(N\big((s, t), B\big) - \widehat{N}\big((s, t), B\big) \Big)^2 \right] \nonumber\\
=\ & \mathbb{E}\left[ \Big(N\big((s, t), B\big) - \Lambda\big((s, t), B\big) \Big)^2 \right] + \mathbb{E}\left[ \Big(\Lambda\big((s, t), B\big) - \widehat{N}\big((s, t), B\big) \Big)^2 \right]  \nonumber\\
\quad &+ \mathbb{E}\left[ \left( N\big((s,t), B\big) - \Lambda\big((s,t), B\big) \right) \left( \Lambda\big((s,t), B\big) - \widehat{N}\big((s, t), B\big)\right)\right]. \label{eq:threeterms}
\end{align}
In this expression, the first term does not depend on $\widehat{N}$, and the third term vanishes due to the definition \eqref{eq:compensator_def} of the compensator. Indeed,
\begin{align*}
&\mathbb{E}\left[ \Big(N\big((s, t), B\big) - \Lambda\big((s, t), B\big)\Big)\Big(\Lambda\big((s, t), B\big) - \widehat{N}\big((s, t), B\big)\Big) \right]  \\
&= \mathbb{E}\left[ \mathbb{E}\left[ \Big(N\big((s, t), B\big) - \Lambda\big((s, t), B\big)\Big)\Big(\Lambda\big((s, t), B\big) - \widehat{N}\big((s, t), B\big)\Big) \Big\vert \mathcal{H}_s \right] \right] \\
&= \mathbb{E}\left[ \Big(\Lambda\big((s, t), B\big) - \widehat{N}\big((s, t), B\big)\Big)  \mathbb{E}\left[ \Big(N\big((s, t), B\big) - \Lambda\big((s, t), B\big)\Big)\Big\vert \mathcal{H}_s \right] \right] =0
\end{align*}
where the second equality follows from $\mathcal{H}_s$-predictability of $\widehat{N}$ and $\Lambda$.
 It suffices therefore to minimise the second term of \eqref{eq:threeterms}, which implies that the target of prediction is ultimately the compensator $\Lambda$. This is intuitive and appealing as we do not predict the remainder $\varepsilon_N = N - \Lambda$, which is the innovations martingale.

The developments so far are not specific to point processes. We will now turn to the more specific question of linear prediction of $N$.

The next sections are laid out as follows. In Section \ref{sec:linear}, we specialise the means-square prediction problem to \emph{linear} prediction, which is tractable and has well-known commonalities with discrete-time theory. The problem can be cast as one of projection onto linear subspaces, and leads to an integral equation with special structure (Theorem \ref{thm:innov}).
In Section \ref{sec:computing_blp}, we discuss computing the linear predictor in several ways. A first method relies on an autoregressive formulation of the linear predictor as the solution to a set of differential equations; we show this is well-suited to the Hawkes process (Section \ref{sec:ar_example}). In the second method, we derive an innovations algorithm and suggest it applies best to Neymann-Scott processes (Section \ref{sec:ma_example}).
To illustrate our proposed methods, we conduct a simulation study in Section \ref{sec:simulations} on two typical examples. Finally, we conclude by suggesting several avenues of further research in line with our findings.

 %in detail the situation where $N$ is approximated by an unbiased linear functional of the past. In this situation, the problem can be cast as one of projection onto linear subspaces. We will therefore use the usual Hilbert space structure of square-integrable random variables and set

%and $P_\mathcal{F}(X)$ will be used to denote the linear projection of the random variable $X$ onto the sigma-algebra $\mathcal{F}$.

%\textcolor{red}{Transition (In this paper..) and table of contents}

\section{Unbiased linear mean-square prediction}
\label{sec:linear}

%\subsection{Smoothness assumption}

If a parametric form for $\Lambda$ is known, then prediction reduces to least-square estimation of the parameters, either analytically or through a numerical scheme. However, the form of the optimal predictor $\Lambda$ is usually not available, since knowledge of $\Lambda$ is close to knowledge of the whole process. In the absence of a known parametric form for $\Lambda$, we propose a linear approximation since it is both simple and computationally tractable.

First, note that so far no hypotheses regarding the compensator $\Lambda$ were needed. In the sequel, we will suppose that the compensator admits a time-intensity.

\begin{assumption}
The compensator $\Lambda$ is absolutely continuous with respect to Lebesgue measure in its first argument, so that it admits a conditional intensity $\lambda$,
\begin{equation*}
\Lambda\left((a, b), B\right) = \int_a^b \lambda(u, B)\, \mathrm{d}u
\end{equation*}
\label{as:intensity}
\end{assumption}

%\begin{remark}
%If Assumption \ref{as:intensity} did not hold, then $\Lambda$ would admit jumps in time which are difficult to predict in the mean-square sense. In the case where the compensator, with respect to time, is a mixed process with both a density and jump component, we expect that a linear approximation will still give reasonable predictions if the jump component is not too large.
%\end{remark}

\begin{remark}
The compensator $\Lambda(t, B)$ predicts $N(t, B)$, and thus its conditional intensity $\lambda(t, B)$ should predict $N(\mathrm{d}t, B)$, which can be seen as taking value $0$ or $1$ depending on whether there is an event at $\left\{t\right\} \times B$.
\end{remark}

%Therefore, we will attempt to predict $\lambda(t, B)$ by a linear predictor $\psi(t, B)$. 

Now that the existence of a time-density of $N$ is assumed, and in view of the remarks in the previous section, we note that the optimisation problem \eqref{prob:min_square} can be written slightly differently,
\begin{align}
\text{ minimise} \quad & 
%\mathbb{E}\left[ \Big(N\big((s, t), B\big) - \widehat{N}\big((s, t), B\big) \Big)^2 \right] 
%C(\psi) = \mathbb{E}\left[ \left(\int_s^t \lambda(u, B) - \widehat{\lambda}(u, B) \,\mathrm{d}u \right)^2 \right],\label{prob:min_square_alt} \tag{BP'}\\
C(\psi) = \left\| \int_s^t \left(\lambda(u, B) - \widehat{\lambda}(u, B) \right)\mathrm{d}u \right\|,\label{prob:min_square_alt} \tag{BP'}\\
\text{s.t.} \quad & \widehat{\lambda} \in \mathcal{P}(\mathcal{H}_s)\nonumber
\end{align}

In this problem, the objective function is the integrated mean-square error, which is significantly more difficult to analyse than the pointwise approximation error of $\lambda(t, B)$. Moreover, there is no way, to our knowledge, to recursively solve the problem when $t$ is increased, that is, the problem has to be solved from scratch for every new value of $t$. % TODO: Check if well explained
As we will see next, much more can be said about the ``pointwise'' optimisation problem of minimising the integrand,
$$
%\mathbb{E}\left[ \Big(\lambda(t, B) - \widehat{\lambda}(t, B) \Big)^2 \right]
\left\|\lambda(t, B) - \widehat{\lambda}(t, B)\right\|
$$
for arbitrary $t \geqslant s$ and $B \in \mathcal{M}$.%, and where $\psi$ is the predictor of the density $\lambda(t, B)$.

\subsection{The linear history}

In order to construct an unbiased linear predictor, we will project $N$ onto the subspace $\mathcal{L}_s$ of $\mathcal{H}_s$-predictable processes which are linear with nonzero mean. To be more specific, we have the following definition which follows \cite{lindquist_optimal_1973}.
\begin{definition}
For every $s > 0$, we define the linear span of the (strict) past history at time $s$ as
%\begin{align}
%\mathcal{L}_s &= \operatorname{span}\left(c + N(u, A) \, \big\vert \, 0 \leqslant u < s, A \in \mathcal{M} \right)
%\end{align}
%Or also
\begin{align}
\mathcal{L}_s &=  \mathrm{closure}\ \mathrm{of}\  \Bigg\{ f_0 + \sum_{i=1}^n f_i \left(N(t_i, A_i) - N(t_i^-, A_i) \right) : n \in \mathbb{N}, f_1, \ldots, f_n \in \mathbb{R}, \\ & \qquad \qquad \qquad 0 \leqslant t_1, \ldots, t_n < s, A_i \in \mathcal{M} \Bigg\}
\end{align}
Any element of $\mathcal{L}_s$ takes the form
\begin{equation}
X = X_0 + \iint_{[0, s) \times \mathbf{M}} f(u, \mu) N(\mathrm{d}u, \mathrm{d}\mu)
\label{eq:lin_form}
\end{equation}
with $f : \mathbb{R} \times \mathbf{M} \to \mathbb{R}$ a function and $X_0$ a real-valued random variable. % \in \mathbb{R}$.
% TODO: Is $f$ here random? Should X_0 be \mathcal{H}_0-measurable?
\end{definition}

Since the linear predictor is the projection $\widehat{\lambda}(t, B | s)$ of $N(t, B)$ onto $\mathcal{L}_s$, it belongs to the latter and is of the form \eqref{eq:lin_form}. Thus, for every $t \geqslant s, B \in \mathcal{M}$, the linear predictor can be written 
\begin{equation}
\widehat{\lambda}(t, B | s) = \widehat{\lambda}_0(t, B) + \iint_{[0, s) \times \mathbf{M}} G(t, u, B, \mu | s) N(\mathrm{d}u, \mathrm{d}\mu),
\label{eq:gfun}
\end{equation}
where $G$ depends on the time $t$ and mark set $B$ at which a prediction is required. The superscript $s$ indicates the limits of the available history. 

%\textcolor{red}{Refer to innovation section}
We will later use (in Section \ref{sec:innov}) a variant of \eqref{eq:gfun} which expresses $\widehat{\lambda}$ as a linear combination of the linear innovations instead,
\begin{equation}
\widehat{\lambda}(t, B | s) = \widehat{\lambda}_0(t, B) + \iint_{[0,s) \times \mathbf{M}} R(t, u, B, \mu | s) \Big(N(\mathrm{d}u, \mathrm{d}\mu) - \widehat{\lambda}(u, \mathrm{d}\mu | u)\, \mathrm{d}u \Big)
\label{eq:pred_innov}
\end{equation}
which is justified because $\widehat{\lambda}(u, A | u) \in \mathcal{L}_s$ since $\mathcal{L}_u$ is a subspace of $\mathcal{L}_s$.

It can be shown that 
\begin{equation}
%G(t, r, B, A | s) = \mathbb{E}\left[\Big(\lambda(t, B) - \psi(t, B | s)\Big)\Big(\lambda(r, A) - \psi(r, A | s)\Big)\right],
G(t, r, B, A | s) = \Big\langle \lambda(t, B) - \widehat{\lambda}(t, B | s), \lambda(r, A) - \widehat{\lambda}(r, A | s) \Big\rangle
\end{equation}
and also
\begin{equation}
G(t, r, B, A | s) = G(r, t, A, B | s).
\end{equation}

By its definition, the linear predictor $\widehat{\lambda}$ is not, in general, an unbiased predictor. However, note that if
\begin{align}
\mathbb{E}\left[\widehat{\lambda}(t, B | s)\right] &= \widehat{\lambda}_0(t, B) + \iint_{[0, s) \times \mathbf{M}} G(t, u, B, \mu | s) \widebar{\lambda} \mathrm{d}u F(\mathrm{d}\mu)  \\
&= \widehat{\lambda}_0(t, B) + \widebar{\lambda} \mathbb{E}\left[\int_{[0, s)} G(t, u, B, M | s)\, \mathrm{d}u\right]
\end{align}
where the random variable $M$ is distributed according to the stationary mark distribution $F$. This implies that by setting
\begin{equation}
\widehat{\lambda}_0(t, B) = \widebar{\lambda}\left(F(B) - \mathbb{E}\left[\int_{[0, s)} G(t, u, B, X | s)\, \mathrm{d}u\right]\right)
\label{eq:pred_intercept}
\end{equation}
the predictor will be made unbiased. In the remainder of this paper, we will work only with unbiased predictors.

%To summarise, we have the following definition:
%\begin{definition}
%A random variable $\psi(t, B)$ is called a linear predictor of $\lambda(t, B)$ (or $N(\mathrm{d}t, B$, see remark \ref{rem:predict_ndt}) if it is the projection of $\lambda(t, B)$ onto the linear history $\mathcal{L}_s$ and thus is written
%\end{definition}

\begin{remark}
Equations \eqref{eq:gfun} and \eqref{eq:pred_innov} represent two views of the linear predictor, one ``autoregressive'' and one ``moving average'' form. In particular, there is a resemblance between \eqref{eq:gfun} and the definition \eqref{eq:marked_hawkes} of a marked Hawkes process. This resemblance can be emphasised by setting $s = t$ and rewriting \eqref{eq:gfun} in terms of $K(t, r, A, B) = G(t, t-r, B, A | t)$.
\end{remark}

%Note that by rewriting $K_t(t-u, A, B) = G_{t, B}(u, A)$ we arrive at a familiar form.
%\begin{definition}
%A linear unbiased predictor $\psi(t, B)$ of $\lambda(t, B)$ is 
%\begin{align}
%\psi(t, B) &= \psi_0(t, B) +  \iint_{[0, s) \times \mathbf{M}} K_t(t-u, \nu, B) N(\mathrm{d}u, \mathrm{d}\nu) \label{eq:lin_pred}, \text{ or } \\
%\psi(t, B) &= \psi_0(t, B) + \sum_{i=1}^{N([0, s) \times \mathbf{M})} K_t(t - T_i, X_i, B) \label{eq:lin_pred_sum}
%\end{align}
%where $\psi_0$ is given by \eqref{eq:pred_intercept} and $K_t$ is a kernel from $\mathbb{R}_+ \times \mathbf{M} \times \mathbf{M}$ to $\mathbb{R}$. % TODO: Check phrasing
%\end{definition}

%The linear predictor can be interpreted as follows: the baseline intensity $\psi_0(B)$ measures the mean rate of events with marks in $B$. The second term captures the feedback or autoregressive effect of the past history on the present value. More specifically, the kernel function $K(u, \nu, B)$ measures the influence on current events with marks in $B$ of an event with mark $\nu$ at a lag $u$ in the past. In the general form of \eqref{eq:lin_pred} and \eqref{eq:lin_pred_sum}, we allow an interaction between the mark values at past and present values. Note the similarity between \eqref{eq:lin_pred} and the definition \eqref{eq:marked_hawkes} of a marked Hawkes process, which we will explore later in this section.

%Even though $\psi_0(t, B)$ is determined by the unbiasedness of $\psi(t, B)$, the remaining unknown $G^s_{t, B}$ is characterised by the 

Finally, as explained earlier, the linear prediction problem we will solve is the pointwise problem
\begin{align}
\text{ minimise}\quad & 
%\mathbb{E}\left[ \Big(\lambda(t, B) - \psi(t, B | s) \Big)^2 \right] \label{prob:min_square_lin} \tag{BLP}\\
\Big\|\lambda(t, B) - \widehat{\lambda}(t, B | s)\Big\|\label{prob:min_square_lin} \tag{BLP}\\
%\text{over} &\  \varPsi_0, \,K \nonumber \\
\text{s.t.} \quad &  \widehat{\lambda} \in \mathcal{L}_s,\ \widehat{\lambda} \text{ unbiased}\nonumber.%, \psi \geq 0 \nonumber
%\label{prob:min_square_lin_2}
\end{align}
where we denote
$
\langle X, Y \rangle = \operatorname{Cov}(X, Y) = \mathbb{E}\left[XY\right] - \mathbb{E}\left[X\right]\mathbb{E}\left[Y\right] \text{ and }
\| X \| = \sqrt{\langle X, X \rangle}.
$

The characteristic property of the best predictor is that it minimises the distance between $\lambda(t, B)$ and $\mathcal{L}_s$, or in other words it is the projection of $\lambda(t, B)$ onto $\mathcal{L}_s$. In the next section, we derive an integral equation corresponding to this fact.

\subsection{Wiener-Hopf-type equation for $G$}

The main result of this section is the following theorem.

\begin{theorem}
Let $\widehat{\lambda}(t, B |s)$ be the best linear unbiased predictor of $N(\mathrm{d}t, B)$, i.e.~the minimiser of \eqref{prob:min_square_lin} in terms of $G$. Then, it obeys the following integral equation
\begin{align}
c(t-r, A, B) = \widebar{\lambda} F(B) G(t, r, B, A | s) + \iint_{(-\infty, s) \times \mathbf{M}} G(t, h, B, \nu | s) c(h-r, A, \mathrm{d}\nu) \opd h
\label{eq:wh-generalised}
\end{align}
for $r < s$ and $A \in \mathcal{M}$, and where $c$ is the covariance measure of $N$.
\label{thm:wh}
\end{theorem}

A variational proof of Theorem \ref{thm:wh} is reported in full to the Appendix, but \eqref{eq:wh-generalised} can also be derived using the orthogonality condition of $\psi$. indeed, by construction, the remainder $N(t, B) - \widehat{\lambda}(t, B | s)$ is orthogonal to the linear history $\mathcal{L}_s$ up to time $s$, and in particular it is orthogonal to $N(r, A)$ for $r < s$ and $A \in \mathcal{M}$. Thus,
\begin{equation}
\Big\langle \lambda(t, B) - \widehat{\lambda}(t, B | s), N(r, A)\Big\rangle = 0
\end{equation}
and
\begin{align*}
&c(t-r, A, B) = \iint_{(-\infty, s) \times \mathbf{M}} G (t, h, \nu, B | s) \,\mathbb{E}\left[ N(r, A) N(\mathrm{d}h, \mathrm{d}\nu) \right] \\
&= \int_{-\infty}^s \int_{-\infty}^r \int_{\mathbf{M}^2} G (t, h, \nu, B | s) \,\mathbf{1}_A(z) \,\mathbb{E}\left[ N(\mathrm{d}u, \mathrm{d}z) N(\mathrm{d}h, \mathrm{d}\nu) \right] \\
&= \int_{-\infty}^s \int_{-\infty}^r \int_{\mathbf{M}^2}  G (t, h, \nu, B | s) \,\mathbf{1}_A(z)\,\Big( \mathrm{d}u \,\mathrm{d}h \,\mathrm{d}z  \, \widebar{\lambda} F(\mathrm{d}z) \delta(u-h) \delta(z-\nu) + c(h-\mathrm{d}u, A, \mathrm{d}\nu) \Big) \\
&= \widebar{\lambda} F(B) G(t, r, B, A | s) + \iint_{(-\infty, s) \times \mathbf{M}} G(t, h, B, \nu | s) \,c(h-r, A, \mathrm{d}\nu) \opd h
\end{align*}

Theorem \ref{thm:wh} extends the ``model independent origin'' of the Hawkes process, as presented in the thesis of \cite{jaisson_market_2015}, to the marked point process case. The key idea is that the marked Hawkes process with parameters $\eta$ and $K$ satisfies the Wiener-Hopf equation (rewritten with $K$ instead of $G$) when $s = t$. Thus, an estimate for the best linear predictor in the instantaneous case $s=t$ can be obtained by fitting a Hawkes process, with the support of $K$ sufficiently large.

A special case of \eqref{eq:wh-generalised} was derived by \citep{bacry_first-_2016} with continuous marks and a separable kernel $G$. The emphasis there is on the Hawkes process and its estimation. While we also note the role of the Hawkes process as a purely autoregressive process (see Section \ref{sec:dl}) the focus here remains on linear prediction in general.

%As mentioned previously, best linear mean-square prediction in the marked point process case shares many similarities with the usual theory of linear prediction for time series. In particular, the conditional expectation plays the role of a projection operator under a scalar product defined by the covariance. The best linear predictor is obtained by projection onto the linear subspace spanned by the history of the process; this leads to a set of orthogonality equations which uniquely characterise the solution. For the analogous development in the discrete-time stochastic process case, see e.g.~\cite{brockwell_time_2006}.

%The equation \eqref{eq:wh-generalised} is an integral equation over $\mathbb{R} \times \mathbf{M} \times \mathbf{M}$, and it is a net generalisation of the Wiener-Hopf equation studied in \cite{hawkes_point_1971}, \cite{bacry_first-_2016}, and others when marks are omitted.% In its full generality, the solution of \eqref{eq:wh-generalised} appears quite difficult.

%In other words, the task of linear unbiased prediction of $N$ can also be seen as fitting a Hawkes process to observations from $N$.
%This situation is analogous to linear mean-square prediction (also known as causal Wiener filtering) of a discrete-time stochastic process, which involves the same equations as when fitting an Gaussian autoregressive (AR) process to the data \citep{brockwell_time_2006}.

Under some hypotheses on the function $G$, one can show that the optimisation problem \eqref{prob:min_square_lin} admits a unique global minimum. More precisely, if $G(\cdot, | s)$ belongs to a reflexive Banach space of functions, then since the objective function of \eqref{prob:min_square_lin} is strictly convex and coercive in $K$, there exists a unique global minimum \citep[for example]{ekeland_convex_1999}. Examples of such reflexive Banach spaces are the Hilbert space $L^2(\mathbb{R} \times \mathbf{M}^2)$, which also plays a key role in the discrete-time theory \citep{brockwell_time_2006}.

It is worth questioning, however, whether the existence of solutions in $L^2$ is relevant if the model is applied to data directly; the conditions for existence of a stationary Hawkes process are quite different, since it is required that $\mathbb{E}[\int_0^{\infty} K] < 1$ and $K \geq 0$. It is conceivable that finding the best linear unbiased predictor of $N$ in some situations will result in an estimated $K$ which does not respect one or both of these conditions, so that it cannot properly be interpreted as a classical Hawkes process. This is similar to an estimated autoregressive process failing to obey the stability condition on the roots of the autoregressive polynomial. On the other hand, there exist modifications of the Hawkes process where $\lambda(t, B)$ is replaced by its positive part $\lambda_+(t, B)$ \citep{costa_renewal_2018}. Moreover, the stability criterion ensures existence over the whole temporal domain $\mathbb{R}$ and it may not be immediately relevant if the analysis focuses rather on a small time interval $[0, T]$.

We have shown in Theorem \ref{thm:wh} that best linear unbiased prediction (in the mean-square sense) of a stationary marked point process results in a particular kind of integral equation, which has links with the Hawkes process. In the next section, we discuss solving this equation through either an auxiliary differential system or through an innovations representation. Both methods of solution are recursive in a sense to be defined subsequently.

% TODO: Explain the role of $K$ as a partial autocorrelation function; ref. B&D Corollary 5.2.1

\section{Computing the best linear predictor}

\label{sec:computing_blp}

After establishing a generalised Wiener-Hopf equation in Theorem \ref{thm:wh}, we now turn to methods of solution. From a statistical perspective, we are most interested in solving the equation for $G$ once $C$ has been estimated from paths of the process. This results in a plug-in estimator for $G$ whose properties depend on both the covariance estimation and the method of solution of the integral equation. %If we estimate the moment measure $c$ involved in the equation, then the solution is an estimator for the prediction kernel $G$. In particular, it can be seen as an estimator of the Hawkes kernel $K$ when the model reduces to \eqref{eq:marked_hawkes}.

\subsection{Existing methods}

The approach outlined above, namely to estimate $C$ and thereby obtain an estimate of $G$ by solving the integral equation, has been applied in several different domains. In spatial statistics, where the marks lie in $\mathbb{R}^2$, it is usually known as Kriging \citep{diggle_statistical_2014}. A relevant reference is the work of \cite{gabriel_estimating_2017}, which describes a spatial methodology that resembles ours.

In time series analysis, the original problem of mean-square linear prediction, as well as the name ``Wiener filter'' is due to \cite{wiener_extrapolation_1949}. The solution of the so-called Yule-Walker equations by the Durbin-Levinson algorithm \citep{brockwell_time_2006,whittle_fitting_1963} is a classical technique with applications to $\text{AR}(p)$, $\text{MA}(q)$, and $\text{ARMA}(p,q)$ processes and selection of the orders $p,q$.% As we noticed in this work, best linear prediction and fitting an autoregressive process are often very similar if not equivalent.

Leaving aside the generality of Theorem \ref{thm:wh}, we will treat the case where the mark space $\mathbf{M}$ is discrete, with marks taking values in $\left\{1, \ldots, d\right\}$. The point process can then be interpreted as vector of univariate point processes,
$$
\mathbf{N}(t) = 
\begin{pmatrix}
N(t, 1), 
N(t, 2),
\cdots,
N(t, d)
\end{pmatrix}^\top =
\begin{pmatrix}
N_1(t), 
N_2(t),
\cdots,
N_d(t)
\end{pmatrix}^\top.
$$
In this specific case the linear predictor takes the form
\begin{equation}
\widehat{\lambda}_i(t) = \widehat{\lambda}_{0,i} + \int_0^s \sum_{j=1}^d G_{ji}(t, u | s) N_j(du)
\label{eq:lin_predictor_multivariate}
\end{equation}
where the $G_{ij}(t, u | s) = G(t, u, j, i | s)$ can be collected into a matrix $\mathbf{G}$.
The process described above is also called the multivariate Hawkes process \citep[e.g.]{laub_hawkes_2015}. The Wiener-Hopf equation then reads
\begin{equation}
\mathbf{c}(t-r) = \mathbf{G}(t, r | s) + \int_0^t \mathbf{G}(t, \tau | s) \,\mathbf{c}(\tau - r)\, \mathrm{d}\tau
\label{eq:wh_simplified}
\end{equation}
where $\mathbf{c}(t)$ is the matrix of covariance measures. 

%Taking for now the more general form \eqref{eq:gfun}, the results from the more general marked case hold particularly, 

\subsection{Autoregressive representation}
\label{sec:dl}

Since work by Kailath and co-authors \citep{kailath_innovations_1970, kailath_note_1972, kailath_new_1973, sidhu_development_1974,sidhu_shift-invariance_1974} as well as \cite{lindquist_optimal_1974, lindquist_fredholm_1975-1} it is known that solution of the continuous-time linear prediction problem (Wiener filtering in white or Gaussian noise) can be effected in a recursive way. In this series of papers, the authors in effect develop a continuous-time analogue to the Durbin-Levinson recursions for solving the Wiener-Hopf equation. The result is a set of differential equations, one of which describes the forward time evolution (from $0$ to $s$) and the other in the reverse direction.% Their results partially apply to our setting, since the noise process which in our case is the innovation martingale $M(\mathrm{d}t)$ has independent increments. 

The integral equation \eqref{eq:wh_simplified} is a Fredholm integral of the second type whose kernel is the covariance function of $N$, and whose data is $C(t-\cdot)$. Its solution $G$ is the linear predictor weighting function of \eqref{eq:gfun}. The resolvent formalism of integral equations (see e.g. \cite{bellman_functional_1957}) expresses $G$ as a function of the resolvent $R$ of $C$ and the data, that is,
\begin{equation}
G(t, r | s) = c(t - r) + \int_0^t R(r, t, u) \mathbf{c}(t-u) \mathrm{d}u
\end{equation}

By Proposition 2.2~of \cite{lindquist_fredholm_1975-1}, we have the Bellman-Krein equation
%the Bellman-Krein equation:
\begin{equation}
\frac{\partial \mathbf{G}}{\partial s}(t, r | s) = - \mathbf{G}(t, s | s) \mathbf{G}(s, r | s)
\end{equation}
Suppose for the remainder of the section that $s = t$, i.e.~the process is observed up to, but not including time $t$. It is then useful to rewrite the above equation in terms of the smoothing functions
$$
\mathbf{F}(t, r) = \textbf{G}(t, t-r | t), \ \textbf{F}^*(t, r) = \textbf{G}(0, r | t)
$$
and the function $\boldsymbol{\Gamma}(t) = \mathbf{G}(t, 0 | t)$. This leads to the forward and backward equations
\begin{align}
\frac{\partial \mathbf{F}}{\partial t}(t, r) &= - \boldsymbol{\Gamma}(t) \mathbf{F}^*(t, t-r) \\
\frac{\partial \mathbf{F}^*}{\partial t}(t, r) &= - \boldsymbol{\Gamma}(t)^\top\mathbf{F}(t, t-r)
\label{eq:fb_diff}
\end{align}
with boundary conditions $\mathbf{F}(t, t) = \boldsymbol{\Gamma}(t)$, $\mathbf{F}^*(t, t) = \boldsymbol{\Gamma}(t)^\top$. The corresponding integral equations are
\begin{align*}
\mathbf{F}(t, r) &= \mathbf{\Gamma}(t) - \int_{r}^t \mathbf{\Gamma}(\tau) \mathbf{F}^*(\tau, \tau - r) \opd \tau\\
\mathbf{F}^*(t, r) &= \mathbf{\Gamma}(t)^\top - \int_{r}^t \mathbf{\Gamma}(\tau)^\top \mathbf{F}(\tau, \tau - r) \opd \tau.
\end{align*}
We also note the useful relationships
\begin{align}
\boldsymbol{\Gamma}(t) &=\mathbf{C}(t) - \int_0^t \mathbf{C}(t-r) \mathbf{F}^*(t, r) \opd r = \left(\mathbf{C}(t) - \int_0^t \mathbf{C}(r) \mathbf{F}^*(t, t-r) \opd r \right)\mathbf{V}^*(t)^{-1}\\
\boldsymbol{\Gamma}^*(t) &= \left(\mathbf{C}(-t) - \int_0^t  \mathbf{C}(r-t) \mathbf{F}(t, r)  \opd r\right)\mathbf{V}(t)^{-1}
\label{eq:integr_gamma}
\end{align} 
where
$$
\mathbf{V}(t) = \mathbf{C}(0) - \int_0^t \mathbf{C}(r) \mathbf{F}^*(t, r) \opd r,
$$
and which derive from the Wiener-Hopf integral equations combined with the definition of $\mathbf{\Gamma}$.
% TODO: Check if the second variant is correct !

Together, all the formulae above may be seen as the continuous-time analogue of \citeauthor{whittle_fitting_1963}'s \citeyearpar{whittle_fitting_1963} algorithm, which extends the Durbin-Levinson algorithm to multivariate processes. A first-order semi-implicit Euler discretisation of the system can be shown to recover Whittle's algorithm. It is conceivable that other numerical schemes applied to the differential system \eqref{eq:fb_diff} will also yield good results; to this end, see simulations in Section \ref{sec:simulations}.

\begin{remark}
The function $\mathbf{\Gamma}(h)$ can be understood as a partial covariance (correlation) function of the process at two times $h$ lags apart, exactly like the partial autocovariance (correlation) function (PACF) in time series analysis. In fact, one has \citep{lindquist_optimal_1974}
\begin{equation}
\mathbf{\Gamma}(h) = \mathbb{E}\left[\left(\boldsymbol{\lambda}(t+h)- \widehat{\boldsymbol{\lambda}}(t+h)\right)\left(\boldsymbol{\lambda}(t)- \widehat{\boldsymbol{\lambda}}(t)\right)\right]  
\end{equation}
To our knowledge, use of such a PACF in point process analysis is  novel and could find interesting applications, particularly when choosing a suitable model.
\end{remark}

\subsection{Example: AR process}
\label{sec:ar_example}

In his PhD thesis, Thibault Jaisson has shown \citep{jaisson_market_2015} that the best linear predictor of an unmarked point process $N$ is related to the Hawkes process, and its parameters depend on the first two moment measures of $N$. He called this the model-independent origin of the Hawkes process, and argued that the Wiener-Hopf equation (which we generalised in Theorem \ref{thm:wh}) directly characterised the Hawkes process. In this example, we will go further and show that the Durbin-Levinson type estimator is perfectly suited to the Hawkes process, since then the smoothing function has support only within a finite interval of fixed length.

%The Durbin-type algorithm is particularly well adapted to the case where $\mathbf{N}$ is a Hawkes process, as the smoothing function concentrates on an interval of fixed length.
\begin{proposition}
Let $N$ be a multivariate Hawkes process, with conditional intensity
\begin{equation}
\boldsymbol{\lambda}(t) = \boldsymbol{\eta} + \int_0^t \mathbf{K}(t-s) \mathbf{N}(\mathrm{d}s).
\end{equation}
and suppose the supports of $\mathbf{K}_{ij}$ are contained in $[0, H_{ij})$. Then the functions $\mathbf{\Gamma}_{ij}$ cut off after $H_{ij}$ in the sense that $\mathbf{\Gamma}_{ij}(h) = 0$ for $h > H_{ij}$.
\label{prop:ar_cutoff}
\end{proposition}
\begin{proof}
In this setting the linear predictor corresponds exactly to the data generating process: both $\boldsymbol{\lambda}$ and $\boldsymbol{\widehat{\lambda}}$ are a linear combination of the past $\mathcal{H}_t$. Indeed, if we write the linear predictor of $\boldsymbol{\lambda}(t)$ based on the history $[t, t-h)$, the corresponding integral equation \eqref{eq:wh-generalised} admits the solution
$$
G(t, t-r | t) = \begin{cases}
K(t-r) &\text{if } 0 \leqslant r < H_{ij} \\
0 &\text{if } H_{ij} < r < h
 \end{cases}
$$
and this solution is unique. Thus 
\begin{align}
\widehat{\boldsymbol{\lambda}}_{j}(t) &= \widehat{\lambda}_{0, j} + \sum_{i=1}^d \int_{0}^{H_{ij}} K_{ij}(t-r) N_i(\mathrm{d}r) \\
&= \widehat{\lambda}_{0, j} + \sum_{i=1}^d \int_{0}^{+\infty} K_{ij}(t-r) N_i(\mathrm{d}r) 
\end{align}
and now we have
%note additionally that if the support of $K_{ij}$ is contained in $[0, H_{ij})$, 
%We have
\begin{align*}
\mathbf{\Gamma}(h) &= \mathbb{E}\left[\left(\boldsymbol{\lambda}(t) - \widehat{\boldsymbol{\lambda}}(t)\right) \left(\boldsymbol{\lambda}(0) - \widehat{\boldsymbol{\lambda}}(0)\right)'\right] \\
&= \mathbb{E}\left[\varepsilon_N (\mathrm{d}t)\left(\boldsymbol{\lambda}(0) - \widehat{\boldsymbol{\lambda}}(0)\right)'\right] = 0
\end{align*}
because of the independent increments of $\varepsilon_N$.
\end{proof}

The cut-off of $\mathbf{\Gamma}$ implies that the recursions \eqref{eq:fb_diff} are computed only for values of $t$ within the support of $K$, allowing for increased computational efficiency.

In the next section, we derive another algorithm for computing the linear predictor, using the innovations representation.

\subsection{Innovations representation}
\label{sec:innov}

%\textcolor{red}{Change N to $\lambda$ below}.
Instead of the autoregressive representation \eqref{eq:lin_predictor_multivariate}, it is also possible to express the linear predictor as a function of the innovations process. Namely, if we set
\begin{equation}
\widehat{\boldsymbol{\lambda}}(t) = \int_0^t \boldsymbol{\Theta}(t, t-u) \left[ \mathbf{N}(\mathrm{d}u)  - \widehat{\boldsymbol{\lambda}}(u) \,\mathrm{d}u\right],
\label{eq:innov_predict}
\end{equation}
then the matrix function $\boldsymbol{\Theta}$ is determined by the following set of recursive equations. 

\begin{theorem}[Innovations algorithm]
Given the innovations predictor \eqref{eq:innov_predict}, the matrix function $\boldsymbol{\Theta}$ and variance function $\mathbf{V}$ satisfy the following equations:
\begin{align}
\boldsymbol{\Theta}(t, h) &= \left(\mathbf{C}(h) - \int_0^{t-h} \boldsymbol{\Theta}(t-h,x) \mathbf{V}(t-h-x) \boldsymbol{\Theta}^\top(t,x+h) \,\mathrm{d}x\right) \mathbf{V}(t-h)^{-1} \label{eq:innov_theta}\\
\mathbf{V}(h) &= \boldsymbol{\Gamma}(0) - \int_0^h \boldsymbol{\Theta}(t,u) \mathbf{V}(h-u) \boldsymbol{\Theta}(t, u)^\top\mathrm{d}u \label{eq:innov_var}
\end{align}
\label{thm:innov}
\end{theorem}

\begin{proof}
Since $\mathbf{\widehat{N}}(t)$ in \eqref{eq:innov_predict} is a projection onto $\mathcal{L}_s$, 
\begin{align*}
&\mathbb{E}\left[\mathbf{N}(\mathrm{d}t)\left(\mathbf{N}(\mathrm{d}u) - \widehat{\boldsymbol{\lambda}}(u) \opd u\right)^\top \right] = \mathbb{E}\left[\widehat{\boldsymbol{\lambda}}(\mathrm{d}t)\left(\mathbf{N}(\mathrm{d}u) - \widehat{\boldsymbol{\lambda}}(u)\,\mathrm{d}u\right)^\top\right]\\
&= \int_0^t \boldsymbol{\Theta}(t, t-x) \mathbf{E}\left[\left(\mathbf{N}(\mathrm{d}x) - \widehat{\boldsymbol{\lambda}}(x)\mathrm{d}x\right)\left(\mathbf{N}(\mathrm{d}u) - \widehat{\boldsymbol{\lambda}}(u)\mathrm{d}u\right)^\top\right] \\
&= \boldsymbol{\Theta}(t, t-u)\, \mathbb{E}\left[\left(\mathbf{N}(\mathrm{d}u) - \widehat{\boldsymbol{\lambda}}(u)\mathrm{d}u\right)\left(\mathbf{N}(\mathrm{d}u) - \widehat{\boldsymbol{\lambda}}(u)\mathrm{d}u\right)^\top\right]
\end{align*}
which we rewrite
\begin{equation}
\boldsymbol{\Theta}(t, t-u) \mathbf{V}(u) = \mathbb{E}\left[\mathbf{N}(\mathrm{d}t)\left(\mathbf{N}(\mathrm{d}u) - \widehat{\boldsymbol{\lambda}}(u) \opd u\right)^\top \right]
\label{eq:orthog_innov_var}
\end{equation}
and in particular when $u = 0$,
$$
\boldsymbol{\Theta}(t, t) \mathbf{V}(0) = \mathbb{E}\left[\mathbf{N}(\mathrm{d}t), \mathbf{N}(\mathrm{d}x)_{x=0}^\top \right] = \mathbf{C}(t) \mathrm{d}t%\textcolor{red}{dt?}
$$
Now substituting \eqref{eq:innov_predict} into \eqref{eq:orthog_innov_var} we find
\begin{align}
&\boldsymbol{\Theta}(t, t-u) \mathbf{V}(u) =
\mathbb{E}\left[\mathbf{N}(\mathrm{d}t)\mathbf{N}(\mathrm{d}u)^\top \right] - \mathbb{E}\left[\mathbf{N}(\mathrm{d}t)\widehat{\boldsymbol{\lambda}}(u)^\top \opd u\right] \\
 &=\mathbb{E}\left[\mathbf{N}(\mathrm{d}t)\mathbf{N}(\mathrm{d}u)^\top \right] - \mathbb{E}\left[\mathbf{N}(\mathrm{d}t)\int_0^u  \left( \mathbf{N}(\mathrm{d}x) - \widehat{\boldsymbol{\lambda}}(x) \opd x\right)^\top \boldsymbol{\Theta}(u, u-x)^\top\right] \opd u \nonumber \\
&=\mathbf{C}(t-u) - \int_0^u  \boldsymbol{\Theta}(t, t-x) \mathbf{V}(x)  \boldsymbol{\Theta}^\top(u, u-x) \opd x \opd u
\end{align}
which yields \eqref{eq:innov_theta} when the substitution $u = t-h$ is carried out.

\end{proof}

%As with Theorem \ref{thm:wh} we leave the full proof to Appendix \ref{app:proof_thm2}. 

As in the autoregressive representation, it is straightforward to discretise equations (\ref{eq:innov_theta}-\ref{eq:innov_var}) and this recovers the classical innovations algorithm in several dimensions, see e.g.~\cite{gomez_multivariate_2016}.

%We have already seen that the autoregressive form of the linear predictor bears striking resemblance to a Hawkes process. This is in line with the discrete-time theory, in which a purely autoregressive process (of order $p$, say) is estimated by at most $p$ Durbin-Levinson coefficients.

%\begin{proof}
  
%\end{proof}

\subsection{Example: MA process}

The innovations form of the best linear predictor is best applied to a process composed purely of `innovations'. In the point process case, this is the Neyman--Scott process, in which innovation events following a Poisson process each generate an inhomogeneous cluster of events. See also the comments in \S2.4 of \cite{wheatley_arma_2018}.

Suppose thus that $\mathbf{N} = (N_1, \ldots, N_d)'$ is a Neyman-Scott process with shot-noise $\Theta$, a $d \times d$ matrix function. The structure is that of a Cox process whose driving intensity is a multivariate Poisson process $\mathbf{N}_0$ (we need not give more details about $\mathbf{N}_0$ for this explanation). Given the latent process $\mathbf{N}_0$, $\mathbf{N}_j$ is formed by attaching a cluster of points with intensity $\Theta_{ij}$ centred at each point of the $i$-th coordinate of $\mathbf{N}_{0}$. Thus, the conditional intensity takes the form
\begin{equation}
\boldsymbol{\lambda}(t) = \int_0^t \mathbf{\Theta}(t-s) \mathbf{N}_0(\mathrm{d}s). \label{eq:ns}
\end{equation}

\begin{remark}
The intensity \eqref{eq:ns} and the Hawkes process intensity are somewhat similar. The difference lies in the points which contribute to the intensity: in the Neyman-Scott process only the latent points generate a cluster, while in the Hawkes process every point becomes the progenitor of its own cluster.
\end{remark}

It can be shown just as in Proposition \ref{prop:ar_cutoff} that the innovations algorithm applied to \eqref{eq:ns} operates within the support of $\mathbf{\Theta}$, which reduces its computational burden. The key is to notice that the covariance measure of $N$ has this same support.

%\noindent An alternative expression for the Neyman-Scott intensity is given in the following result.
%\begin{proposition}
%Let $\mathbf{N}$ be a Neyman-Scott process with parameter $\mathbf{\Theta}$. Then it admits the conditional intensity
%\begin{equation}
%\boldsymbol{\lambda}(t) = \int_0^t \mathbf{\Theta}(t-s) \left[ \mathbf{N}(\mathrm{d}s) - \boldsymbol{\lambda}(s) \mathrm{d}s \right]% + \mathbf{N}_0(\mathrm{d}t)
%\end{equation}
%\end{proposition}
%\begin{proof}

%Let $C_i$ be the point process representing the cluster generated by the $i$-th background point $\tau_i$ of $N_0$. Then we may write
%\begin{equation}
%N(t) = N_0(t) + \sum_{i=1}^{N_0(t)} C_i(t)
%\end{equation}
%If we first condition on the joint history $\mathcal{H}_{(N, N_0)}(t)$ of $N$ and $N_0$ up to $t$, then
%\begin{equation}
%\lambda(t) = \eta + \sum_{i=1}^{N_0(t)} \Theta(t-\tau_i)
%\end{equation}

%\end{proof}

\subsection{Estimators}

As mentioned previously, both the Durbin-Levinson algorithm of Section \ref{sec:dl} and the Innovations algorithm of Section \ref{sec:innov} generalise the discrete-time algorithms they are named after. When analysing point process data, a simple approach is then to discretise the real line, bin the data, and apply the corresponding discrete-time algorithm. % This renders much of the current section moot and the algorithms have already been implemented, for example in \texttt{R} \citep{r_core_team_r_2022}. 
However, binning only makes sense if the discretisation is fine enough, as explained in \cite{kirchner_estimation_2017}. Such a fine discretisation results in very large sparse systems, however, and this may cancel the advantage gained from using recursive estimation methods.

Instead, another choice is to remain in the continuous-time domain and estimate $C$ accordingly. Then, either of the two previously mentioned methods can be used in conjunction with a differential equation solver to obtain a solution in terms of $\mathbf{F}$ or $\mathbf{\Theta}$.

In spatio-temporal processes, $C$ is most commonly estimated using a parametric kernel function. Usually, the variogram, a variant on the covariance measure, is estimated instead, but this does not impact our work. If $C$ is defined using a parameter $\beta$, say $C = C_\beta$, we could for instance estimate $\beta$ by least-squares. Common parametric forms for $C$ include the Brownian, Mat\'{e}rn and squared exponential kernels. 

Finally, we may choose to estimate $C$ non-parametrically, but it is not clear how to do this for general marks. In the case of a spatial process on the plane, some authors \citep{huang_nonparametric_2011} have proposed a spectral approach, but this does not generalise in a straightforward way to more complex marks.

\section{Numerical illustration}
\subsection{Simulations}
\label{sec:simulations}

In this section, we investigate the performance of our proposed method on two idealised examples. 

We first mention the work of \cite{bacry_first-_2016}, who pioneer  estimation of the Hawkes kernel by solving the Wiener-Hopf integral equation. They solve a discretised version of the latter using a Nystr\"{o}m method after estimating the conditional law densities in a non-parametric fashion. This method also obtains the kernel $K$ via solution of the Fredholm integral equation, but the method of inversion is more general and does not incorporate the Toeplitz character of the operator involved. In fact, the Durbin-Levinson algorithm can be seen as an improved method for factorising Toeplitz matrices specifically, which further justifies its use beyond the analogies developed in Section \ref{sec:dl}. In the numerical examples, we have also implemented a ``direct inversion'' method in the spirit of \cite{bacry_first-_2016}. We emphasise, however, that this paper aims to discuss more general parallels with discrete-time linear prediction and not specifically the estimation of the Hawkes process, for which many methods have been proposed. %the approach taken in this paper is specific to linear prediction with many parallels to the discrete-time theory as discussed in Section \ref{sec:computing_blp}. %\ref{s A first improvement over existing methods for computing the best linear predictor is then to discretise the continuous-time relationships after estimating $C$ and solve the matrix equations using the Durbin-Levinson method. The solution could be found on the order of $O(p^2)$ operations rather than $O(p^3)$, which is a significant improvement.

We assess the performance gains of our proposed approach over a direct inversion approach. To solve the differential system \eqref{eq:fb_diff}, we propose to use several simple numerical schemes in addition to Whittle's method: the forward and backward Euler methods, Runge-Kutta's second-order method and the midpoint method. The average squared errors in recovering $K$ ($L^2$ error) over 300 runs are shown in the tables below; in the first case, the true covariance function $c$ is used and the error is due only to the numerical scheme. In the second, we use an estimator for $C$ due to \cite{brillinger_estimation_1976}. Average runtimes are also reported over the 300 simulations for each example. 

\paragraph{Example 1: Exponential kernel}

Let $N$ be the unmarked Hawkes process with exponential kernel $K(t) = \alpha e^{-\beta t}$. This is the simplest possible version of the Hawkes process.
The covariance measure $c(t)$ was derived by \cite{hawkes_point_1971} and reads
$$
c(t) = \frac{\lambda \alpha (2\beta - \alpha)}{2(\beta - \alpha)} e^{-(\beta - \alpha)t}
$$
with an additional atom of mass $\lambda$ at the origin. 

The parameters chosen are $\alpha = 1.5$, $\beta = 2$, $\eta = 1$ and the process is observed over $[0, 500]$.  The simulation results are shown in Table \ref{tab:sim1} below. We immediately notice that most methods yield similar results when applied to the estimated covariance measure. By comparison, when the true covariance function is provided, the error is smaller by several orders of magnitude but varies between the different methods. The fully implicit backward Euler method is slower due to the inversion of a linear system at each step. Second-order methods also show the cost of extra steps by being somewhat slower than the first-order methods.

\begin{table}[ht]
\begin{tabular}{lrrrrrr}
Method & Inversion & Whittle & Forward E. & Backward E. & RK2 & Midpoint\\
\hline
Estimated ($10^{-2}$) & 7.44 & 7.44 & 7.32 & 7.43 & 7.58 & 7.45\\
True ($10^{-5}$) & 1.38 & 1.38 & 17.28 & 16.05 & 6.44 & 2.43\\
Time taken (ms) & 4.13 & 5.14 & 5.00 & 346.72 & 8.93 & 8.64
\end{tabular}
\label{tab:sim1}
\caption{Simulation results on Example 1}
\end{table}

%\paragraph{Example 1: Mat\'{e}rn kernel}

%We consider a point process $N$ defined 

\paragraph{Example 2: Hawkes process with inhibitory sinusoidal kernel}

The second example concerns a Hawkes process $N$ with kernel function
$$
K(t) = e^{-3t} \sin(6t).
$$
which takes negative values in a small interval (see Figure \ref{fig:sinexp} below).
\begin{figure}[ht]
\begin{center}
\begin{tikzpicture}
\begin{axis}[axis lines=middle,ymax = 1,xmax = 2]
  \addplot[domain=0:2, samples = 1000] {e^(-3*x) * sin(360/(2*pi) * 6*x)};
\end{axis}
\end{tikzpicture}
\label{fig:sinexp}
\caption{Inhibitory sinusoidal kernel of Example 2}
\end{center}
\end{figure}
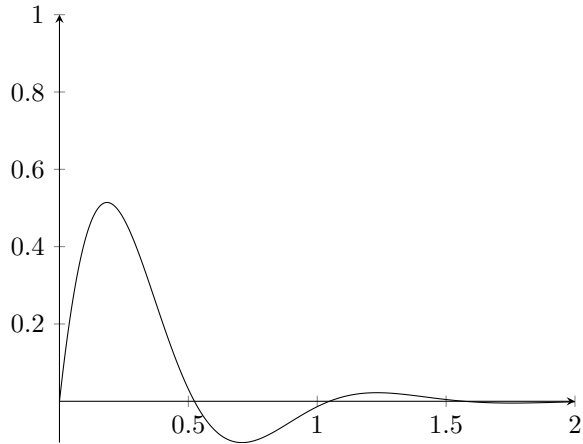

\begin{table}[ht]
\begin{tabular}{lrrrrrr}
& Inversion & Whittle & Forward E. & Backward E. & RK2 & Midpoint\\
\hline
Estimated ($10^{-2}$) & 2.29 & 2.39 & 2.37 & 2.39 & 2.40 & 2.39\\
True ($10^{-4}$) & 1.68 & 1.58 & 1.65 & 1.57 & 1.49 & 1.54 \\ \hline
Time taken (ms) & 7.34 & 37.49 & 22.65 & 1345.71 & 41.77 & 64.22
\end{tabular}
\label{tab:sim2}
\caption{Simulation results on Example 2}
\end{table}

Laplace transforms show that this corresponds to the covariance function
$$
c(t) = \sqrt{\frac{6}{5}} e^{-3t} \sin(\sqrt{30}t)
$$
and the performance of our proposed method is similar (see Table \ref{tab:sim2}).

%\subsection{Application to French polling data}
%\label{sec:application}

In both cases, we have seen that first-order methods usually give good approximations in reasonable time. While direct inversion can be used as a convenient baseline for the examples shown, it is clear that for more complex excitation functions, we will require a finer grid on which to compute the predictor. Because direct inversion (using e.g. the QR decomposition) has cubic complexity in the number of bins, the other methods will begin to outperform it.

\section{Conclusion}

We have described linear prediction of a marked point process, and in doing so, we saw that classical concepts of linear prediction and orthogonality can be easily used in the point process setting.

While we gave some rough ideas for the practical application of linear filtering to marked point processes, it is clear that an analysis of the statistical properties of our estimators is much desired. Additionally, the case of continuous marks must be elucidated with respect to the recursive algorithms described in Section \ref{sec:computing_blp}.
 
In the context of linear prediction, two models -- the Hawkes and Neyman-Scott processes -- played key roles, similar to the complementary roles played by $\text{AR}(p)$ and $\text{MA}(q)$ in time series analysis. It would therefore be worth pursuing the idea of an $\text{ARMA}$ point process, which could become a standard model for point process data. A good amount of progress has been made by \cite{wheatley_arma_2018}, who describes an EM algorithm for estimation of what he calls the ARMA point process, but several characteristics remain unclear, namely the spectral properties of the process and the construction of prediction intervals.

Finally, we did not discuss a frequency-domain approach. This has been set out in some detail in \cite{daley_introduction_2003-1}, and recursive estimation does not play the important role it does in a time-domain approach.

\clearpage

\bibliographystyle{plainnat}
\bibliography{Hawkes}
%\clearpage

\section{Appendix}

\subsection{Proof of Theorem \ref{thm:wh}}
\label{app:proof_thm1}

%\subsection{}
Recall that up to a constant,
$$
\left\| \widehat{\lambda}(t, B) - \lambda(t, B) \right\|^2 = \left\| \widehat{\lambda}(t, B) \right\|^2 - 2 \left\langle \widehat{\lambda}(t, B), \lambda(t, B) \right\rangle
$$
and that we may write $\widetilde{N} = N - \mathbb{E}[N]$ the centred version of $N$ (and similarly with the other processes). Suppressing in the notation for now the dependence on $s$, we then find that
\begin{align*}
&\| \widehat{\lambda}(t, B) \|^2 = \mathbb{E}\left[\left(\widehat{\lambda}_0(B) + \int_{-\infty}^s \int_\mathbf{M} G(t, u, B, \mu) N(\mathrm{d}u, \mathrm{d}\mu) - \widebar\lambda F(B)\right)^2\right] \\
&= \mathbb{E}\left[\left(\int_{-\infty}^s \int_\mathbf{M} G(t, u, B, \mu) \widetilde{N}(\mathrm{d}u, \mathrm{d}\mu)\right)^2\right]\\
&= \mathbb{E}\left[\int_{-\infty}^s \int_\mathbf{M}  \int_{-\infty}^s \int_\mathbf{M} G(t, u, B, \mu) G(t, v, B, z) \widetilde{N}(\mathrm{d}u, \mathrm{d}\mu) \widetilde{N}(\mathrm{d}v, \mathrm{d}z) \right] \\
&= \int_{-\infty}^s \int_\mathbf{M} \int_{-\infty}^s \int_\mathbf{M} G(t, u, B, \mu) G(t, v, B, z) \ \mathbb{E}\left[\widetilde{N}(\mathrm{d}u, \mathrm{d}\mu)\widetilde{N}(\mathrm{d}v, \mathrm{d}z)\right] \\
&= \int_{-\infty}^s \int_\mathbf{M} \int_{-\infty}^s \int_\mathbf{M} G(t, u, B, \mu) G(t, v, B, z) \left[\mathrm{d}u \,C_2(d(u-v), d\mu, dz) +\widebar{\lambda} F(d\mu) \delta(\mu-z) \delta(u-v) du\right] \\
&= \int_{-\infty}^s \int_\mathbf{M} \int_{-\infty}^s \int_\mathbf{M} G(t, u, B, \mu) G(t, v, B, z) C_2(d(u-v), d\mu, dz) + \int_{-\infty}^s \int_\mathbf{M} \widebar{\lambda} G(t, u, B, \mu)^2 F(\mathrm{d}\mu)  \mathrm{d}u
\end{align*} 
by stationarity.
For the second term,
\begin{align*}
\left\langle \widehat{\lambda}(t, B), \lambda(t, B) \right\rangle &= \left\langle \widetilde{\widehat{\lambda}}(t, B), \widetilde{\lambda}(t, B) \right\rangle \\
&= \mathbb{E}\left[\int_{-\infty}^s \int_\mathbf{M} G(t, u, B, \mu)  \widetilde{\lambda}(t, B) \widetilde{N}(\mathrm{d}u, \mathrm{d}\mu) \right] \\
&= \int_{-\infty}^s \int_\mathbf{M} G(t, u, B, \mu)\, C_2(t-u, \mathrm{d}\mu, B) \, \mathrm{d}u.
\end{align*}
Note that since $t<u$ there is no atomic part unlike in the previous set of equations.

We see that any minimiser of the mean-square error, and the best linear predictor in particular, also minimises the following functional:
\begin{align*}
J(G) &:= \int_{-\infty}^s \mathrm{d}u \int_\mathbf{M} G(t, u, B, \mu) \\
&\Big[\int_{-\infty}^s \int_\mathbf{M} G(t, v, B, z) C_2(u-v, d\mu, dz) \opd v+ \widebar{\lambda} G(t, u, B, \mu)^2 F(\mathrm{d}\mu) - 2C_2(t-u, \mathrm{d}\mu, B)
\Big]
\end{align*}
\noindent The minimiser $G^*$of the above quantity necessarily verifies
$$
\mathrm{d}_{\gamma}J(G)\big\vert_{G=G^*} = 0
$$
where $\mathrm{d}_{\gamma}F$ is the Gateaux derivative of a functional $F$ in the direction $\gamma$. Applying this to $J$ with a generic test function $\gamma$, we find
\begin{align*}
&J(G^* + p\gamma) - J(G^*) = p^2 \int_{-\infty}^s \mathrm{d}u \int_\mathbf{M}\gamma(t, u, B, \mu)\int_{-\infty}^s \int_\mathbf{M} \gamma(t,v,B,z) C_2(u-v,\mathrm{d}\mu,\mathrm{d}z) \opd v\\
&\quad+ p^2 \int_{-\infty}^s \mathrm{d}u\int_\mathbf{M}\widebar{\lambda} F(\mathrm{d}\mu) \gamma(t, u, B, \mu) \, \gamma(t,u, B,\mu) \\
&\quad-2p \int_{-\infty}^s\mathrm{d}u \int_\mathbf{M} \gamma(t, u,B, \mu) C_2(t-u, \mathrm{d}\mu, B)\\
&\quad+ p\int_{-\infty}^s \mathrm{d}u\int_\mathbf{M} \gamma(t, u, B, \mu) \int_{-\infty}^s \int_\mathbf{M} G^*(t, v, B, z) C_2(u-v, \mathrm{d}\mu, \mathrm{d}z) \opd v\\
&\quad+ p\int_{-\infty}^s \mathrm{d}u\int_\mathbf{M} \gamma(t,u, B, \mu)  \widebar{\lambda} F(\mathrm{d}\mu) G^*(t, u, B, \mu) \\
&\quad + p\int_{-\infty}^s \mathrm{d}u\int_\mathbf{M} G^*(t, u, B, \mu) \int_{-\infty}^s \int_\mathbf{M} \gamma(t,v, B, z) C_2(u-v,\mathrm{d}\mu,\mathrm{d}z) \opd v \\
&\quad + p \int_{-\infty}^s \mathrm{d}u\int_\mathbf{M}G^*(t,u, B, \mu) \widebar{\lambda}F(\mathrm{d}\mu) \gamma(t,u,B,\mu)
\end{align*}
The first two terms are $o(p)$, and gathering together the other terms we find
\begin{align*}
\frac{J(G^* + p\gamma) - J(G^*)}{p} &= 
\int_{-\infty}^s \mathrm{d}u \int_\mathbf{M}\gamma(t,u, B, \mu) \bigg[ 2 \widebar\lambda F(\mathrm{d}\mu) G^*(t,u,B,\mu) - 2 C_2(t-u,\mathrm{d}\mu,B) \\
& + 2 \int_{-\infty}^s \int_\mathbf{M} G^*(t,v, B, z) C_2(u-v,\mathrm{d}\mu,\mathrm{d}z)\opd v \bigg]  + o(1).
\end{align*}
Since the test function $\gamma$ is arbitrary, a necessary condition for an optimal solution $G^*$ is that the integrand in square brackets vanishes identically, that is,
\begin{equation}
C_2(t-u, \mathrm{d}\mu, B) = \widebar{\lambda} F(\mathrm{d}\mu) G^*(t, u, B, \mu) + \int_{-\infty}^s \int_\mathbf{M} G^*(t, v, B, z) C_2(u-v, \mathrm{d}\mu, \mathrm{d}z) \opd v
\label{eq:wh-marked}
\end{equation}
or in integral form,
\begin{equation}
C_2(t-u, A, B) = \widebar{\lambda} F(A) G^*(t, u, B, A) + \int_{-\infty}^s \int_\mathbf{M} G^*(t, v, B, z) C_2(u-v, A, \mathrm{d}z) \opd v
\label{eq:wh-marked-2}  
\end{equation}
for every $u > 0$ and $A \in \mathcal{M}$.

\subsection{Discretisation of the forward differential system}

In \cite{whittle_fitting_1963}, the best linear predictor coefficients $\mathbf{A}_{p,1}, \ldots, \mathbf{A}_{p,k}$ of order $p$ for a multivariate process $\mathbf{x}_t$ are found by solving the forward and backward Yule-Walker equations,
\begin{align}
\sum_{k=0}^p \mathbf{A}_{p,k} \mathbf{\Gamma}_{j-k} &= 0\  (j = 1, \ldots, p)\\
\sum_{k=0}^p \mathbf{A}_{p,k} \mathbf{\Gamma}_{-k} &= V
\end{align}
and
\begin{align}
\sum_{k=0}^p \mathbf{A}^*_{p,k} \mathbf{\Gamma}_{k-j} &= 0 \ (j = 1, \ldots, p)\\
\sum_{k=0}^p \mathbf{A}^*_{p,k} \mathbf{\Gamma}_k &= V^\top.
\end{align}
The solution is recursively determined as 
\begin{align}
\mathbf{A}_{p+1,k} &= \mathbf{A}_{p,k} + \mathbf{A}_{p+1, p+1} \mathbf{A}^*_{p, p-k+1} \label{eq:recurs_matr} \\  
\mathbf{A}^*_{p+1,k} &= \mathbf{A}^*_{p,k} + \mathbf{A}^*_{p+1, p+1} \mathbf{A}_{p, p-k+1}
\end{align}
for $ k = 1, \ldots, p$ and 
\begin{align}
\mathbf{A}_{p+1, p+1} &= -\Delta_p \left(\mathbf{V}_p^*\right)^{-1} \label{eq:err_mat1} \\
\mathbf{A}^*_{p+1, p+1} &= -\Delta^*_p \left(\mathbf{V}_p\right)^{-1}. \label{eq:err_mat2}
\end{align}
where $\Delta_p = \sum_{k=0}^p \mathbf{A}_{p,k} \Gamma_{p-k+1}$ and $\mathbf{V}_p = \sum_{k=0}^p \mathbf{A}_{p,k} \Gamma_{-k}$.

We will show, by analogy with \cite{whittle_fitting_1963} that a first-order discretisation of the integral equation
\begin{equation}
K \star C = 0,  
\end{equation}
for prescribed $C$ and where $K(0) = I$, obeys these same equations. This compact representation of the Wiener-Hopf equation is obtained by modifying $C$ and $K$ to have atoms at the origin.

To begin, note that
$$
\frac{\partial K}{\partial t}(t, s) = \Gamma(t) K^*(t, t-s)
$$
is discretised to
$$
K_{t+1, s} - K_{t, s} = \Gamma_t K^*_{t, t-s} = K_{t+1,t+1} K^*_{t, t-s}
$$
which corresponds to \eqref{eq:recurs_matr} and similarly with $K_{t+1,s}^*$. To establish (\ref{eq:err_mat1}-\ref{eq:err_mat2}) suppose that a solution $K(t, s)$ is known for $0 \leqslant s \leqslant t$, such that
$$
\int_0^t K(t, u) C(s - u) = 0
$$
holds for $0 \leqslant s \leqslant t$. Then the equation on a larger interval $[0, t+h)$ can be written
\begin{align}
&\int_0^{t+h} K(t+h, u) \,C(t+h-u) \opd u\\
&= \int_0^t K(t+h, u) \, C(t+h-u) \opd u + \int_t^{t+h} K(t+h, u) \, C(t+h-u) \opd u = 0 \label{eq:partition}
\end{align}
Substituting the differential equation for $K$ into \eqref{eq:partition} then yields
$$
K(t+h, t+h) \left[ \int_0^t K^*(t, t-u) C(t+h-u) \opd u\right] = -\int_0^t K(t, u) C(t+h - u) \opd u
$$
which matches $\eqref{eq:err_mat1}$.

\label{app:discretised_system}
\end{document}